%% file: main.tex
\documentclass[journal]{IEEEtran}
\input{edits}

\usepackage{amsmath}
\usepackage{amsfonts}
\usepackage{amssymb}
\usepackage{amsthm}
\usepackage{dsfont}
\usepackage{cite}
\usepackage{blindtext,graphicx}
\usepackage{colortbl}
\usepackage{threeparttable}
\usepackage{subfigure}
\usepackage{siunitx}
\usepackage{mathtools}
\DeclareSIUnit[]{\pu}{p.u.}
\DeclareSIUnit[]{\VA}{VA}

\newtheoremstyle{bfnote}%
{}{}%
{\itshape}{}%
{\bfseries}{.}%
{ }%
{\thmname{#1}\thmnumber{ #2}\thmnote{ (#3)}}
\theoremstyle{bfnote}
\newtheorem{thm}{Theorem}

\DeclareSymbolFont{bbold}{U}{bbold}{m}{n}
\DeclareSymbolFontAlphabet{\mathbbold}{bbold}

\usepackage{enumitem}
\setlist[enumerate]{leftmargin=*}
\setlist[itemize]{leftmargin=*}

\ifCLASSINFOpdf
\else
\fi
\usepackage{amsmath}
\begin{document}

\title{Storage-Based Frequency Shaping Control}
%

\author{Yan Jiang, Eliza Cohn, Petr Vorobev,~\IEEEmembership{Member,~IEEE}, and Enrique Mallada,~\IEEEmembership{Senior Member,~IEEE}
}
\renewcommand{\baselinestretch}{.95}


\maketitle
 
\begin{abstract}
With the decrease in system inertia, frequency security becomes an issue for power systems around the world. Energy storage systems (ESS), due to their excellent ramping capabilities, are considered as a natural choice for the improvement of frequency response following major contingencies. In this manuscript, we propose a new strategy for energy storage -- \emph{frequency shaping control} -- that allows to completely eliminate the frequency Nadir, one of the main issue in frequency security, and at the same time tune the rate of change of frequency (RoCoF) to a desired value. With Nadir eliminated, the frequency security assessment can be performed via simple algebraic calculations, as opposed to dynamic simulations for conventional control strategies. Moreover, our proposed control is also very efficient in terms of the requirements on storage peak power, requiring up to $40\%$ less power than conventional virtual inertia approach for the same performance.   
\end{abstract}


\IEEEpeerreviewmaketitle


\section{Introduction}

The reduction of system inertia, caused by the replacement of conventional synchronous generation with renewable energy sources, is one of the biggest challenges for frequency control in power systems \cite{milano2018}. Lower inertia causes larger frequency deviations during transients, even if the system has adequate primary reserves to keep the steady-state frequency deviation within acceptable limits \cite{puschel2017mapping}. The so-called frequency Nadir -- the lowest value of the frequency during transients -- can become unacceptable for low-inertia systems, which in turn is sometimes regarded as the main reason for limiting further increase of renewable generation penetration \cite{eirgrid2010all, o2014studying}. Fortunately, the recent advancements in power electronics and electric storage technologies provide the potential to mitigate this issue through the use of inverter-interfaced storage units that can provide additional frequency response. With proper controllers, fast inverter dynamics can ensure the rapid response from storage devices. 

A straightforward control approach for energy storage systems (ESS) is to let energy storage units provide simple proportional power-frequency response similar to conventional synchronous generators \cite{sanchez2019security}. However, unlike synchronous generators that produce a delayed response to the control signal, the response of storage units is almost instantaneous. This can help arrest the frequency drop during the first few seconds following a disturbance, while generator turbines are gradually increasing their power output. Moreover, because of the absence of delays, smaller droop coefficients (larger gains) are accessible for energy storage units, which makes them even more efficient during sudden frequency disturbances \cite{greenwood2017frequency}. Thus, an impressive \SI{472}{\mega\watt} of storage has been reported to participate in the frequency response during the recent blackout in the Great Britain system on August $9$, $2019$ \cite{eso2019technical}. A drawback of this droop control strategy is that the storage units will continue to provide their response as long as the system frequency is away from its nominal value, which can lead to rather high requirements on storage capacity.  

Another common control approach is the so-called ``synthetic inertia" (also referred to as ``virtual inertia" (VI)), where energy storage units imitate the natural inertial response of synchronous machines, thus compensating for the lack of physical inertia \cite{fang2017distributed,tielens2016relevance}. Such a control strategy is especially efficient in reducing the frequency Nadir as well as the initial RoCoF, following sudden power imbalances. The topic is widely discussed in literature. We will provide a brief survey of the most relevant sources, yet a comprehensive review is given in \cite{tamrakar2017virtual}. There are various approaches for VI implementation: by wind turbines \cite{arani2012implementing}, by electric vehicles \cite{almeida2015electric}, by distributed energy resources \cite{guggilam2018optimizing}, and by controlling DC-side capacitors of grid-connected power converters \cite{fang2017distributed}. In a recent paper  \cite{poolla2019placement} an important question of VI placement is discussed. 
Finally, we note that both synthetic inertia (derivative control) and droop (proportional control) can be combined into a single control strategy. Sometimes, it is this combined strategy that is referred to as VI.

It is evident that, in many of the power systems around the world, storage facilities can become the main tool for executing frequency control, especially following contingencies, where speed or response is of vital importance. While both synthetic inertia and droop response can be rather effective in improving the frequency transient performance, energy storage units have the potential of implementing a much wider class of control strategies. A high level goal for such strategies would be to provide certain frequency response while minimizing the cost of storage units. The later is mostly determined by the energy and power capacity of storage units required to execute certain strategy. In the present manuscript, we develop a novel control strategy -- \emph{frequency shaping control} -- that guarantees frequency transients without Nadir, while at the same time keeping RoCoF and steady-state frequency deviation within pre-specified limits. We emphasize that eliminating the frequency Nadir means much more than just improving the transient frequency response: it allows to completely change the frequency security assessment procedure by reducing it to simple algebraic operations, rather than dynamic simulations. 

The main contributions of the manuscript are as follows: 

\begin{enumerate}
    \item We analyze the performance of traditional VI control and show its drawbacks, especially in terms of excessive control effort required from the storage.
    
    \item We propose a new control strategy for storage -- \emph{frequency shaping control} -- that allows to turn the system frequency dynamics into a first-order one, thus eliminating frequency Nadir. We show that this strategy requires up to $40\%$ less storage power capacity compared to conventional VI. 
    
    \item We generalize our control strategy for multi-machine and multiple-area systems with arbitrary governor models and show how it can be tuned to satisfy constraints on RoCoF and steady-state frequency deviation.
\end{enumerate}

\section{Modelling Approach and Problem Statement}\label{Sec:modelling}

\subsection{System Model}

We start by considering dynamics of a center of inertia (COI) of a single-area power system so that the whole system can be modelled as an equivalent synchronous machine. Such a representation is proven to be sufficiently accurate for many practical systems \cite{o1999modelling,lalor2005impact,ulbig2014impact}. The generalization to multi-machine and multiple-area representation will be described in Section~\ref{sec:multiple}. Frequency dynamics of such a system can be described by the conventional swing equation:
\begin{equation}\label{eq:freq_exact}
  \frac{2H}{\Omega_0} \dot{\Omega}=P_\mathrm{m}-P_\mathrm{L} +P_\mathrm{b} \,, 
\end{equation}
where $H$ is the combined inertia constant of the system (in second), $\Omega$ is the system frequency (in radian per second), $\Omega_0:=2\pi F_0$ is the nominal frequency ($F_0=\SI{50}{\hertz}$), $P_\mathrm{m}$ is the total mechanical power supplied to the system (in per unit), $P_\mathrm{L}$ is the total load demand (in per unit), and $P_\mathrm{b}$ is the total power supplied by the storage system (in per unit), which also includes any frequency control actions. In this manuscript, we will be mostly interested in dynamics of the system \eqref{eq:freq_exact} subject to a sudden power imbalance $\Delta P$. 


In order to study the frequency dynamics, it is convenient to consider the deviations of all the variables from their equilibrium values. Thus, we will denote as $\omega$ the per unit deviation of frequency from its nominal value, i.e., 
\begin{equation}\label{eq:omega_definition}
    \omega = \frac{\Omega-\Omega_0}{\Omega_0}\,.
\end{equation}
Likewise, $p_\mathrm{m}$, $p_\mathrm{L}$, and $p_\mathrm{b}$ will be used to denote the per unit variations of mechanical power input, electric power demand, and storage power output from their respective nominal values. With such denotations, the frequency dynamics of the system under study can be described by the following equations: 
\begin{subequations}\label{eq:dyn-onebus}
\begin{align}
    \dot{\theta} =&\ \Omega_0 \omega\;,\label{eq:dyn_onebus_phase} \\
    2H \dot{\omega} =&\ p_\mathrm{m} - p_\mathrm{L} -  \alpha_\mathrm{L}\omega + p_\mathrm{b}\;,\label{eq:dyn_onebus_freq}\\
    \tau_\mathrm{T} \dot{p}_\mathrm{m} =& -p_\mathrm{m} - \alpha_\mathrm{g} \omega - K_\mathrm{I}\frac{\theta}{\Omega_0}\;, \label{eq:tur-dyan}\\
      \dot{E}_\mathrm{b} =&\ p_\mathrm{b}\,,
\end{align}
\end{subequations}
where $E_\mathrm{b}$ is the energy supplied by storage and $\theta$ is an auxiliary variable used for the secondary frequency control. The parameters in \eqref{eq:dyn-onebus} are defined as follows: $\tau_\mathrm{T}$ --  turbine time constant (in second), $\alpha_\mathrm{L}$ -- the load-frequency sensitivity coefficient  (in per unit), $\alpha_\mathrm{g}$  -- the aggregate inverse droop of generators (in per unit), and $K_\mathrm{I}$ -- the aggregate secondary frequency control gain of the system (in per unit per second).

The model described by \eqref{eq:dyn-onebus} is shown by a block diagram in Fig.~\ref{fig:1bus-model}. The conventional generator block (with primary and secondary controls) is shaded in blue. We will denote the aggregate transfer function of this block (in Laplace domain) as $\hat{g}(s)$. For clarity of our derivations, we use a simplified first-order turbine representation. The generalization to more complex models will be provided in Section \ref{sec:multiple}.

Compared to conventional generators, inverter-interfaced storage have much faster dynamic response rates (few decades of milliseconds) that allow for more control flexibility. Thus, at the timescales of frequency dynamics, we can assume that storage can provide any shape of power response (within the installed capacity capability). We denote the storage frequency response function as $\hat{c}(s)$, i.e.,  
\begin{equation}\label{eq:inverter-tf}
\hat{p}_\mathrm{b}(s)=\hat{c}(s)\hat \omega(s)\,.
\end{equation}
The detailed form of $\hat{c}(s)$ depends on a chosen control strategy. 


\begin{figure}[!t]
\centering
\includegraphics[width=0.65\columnwidth]{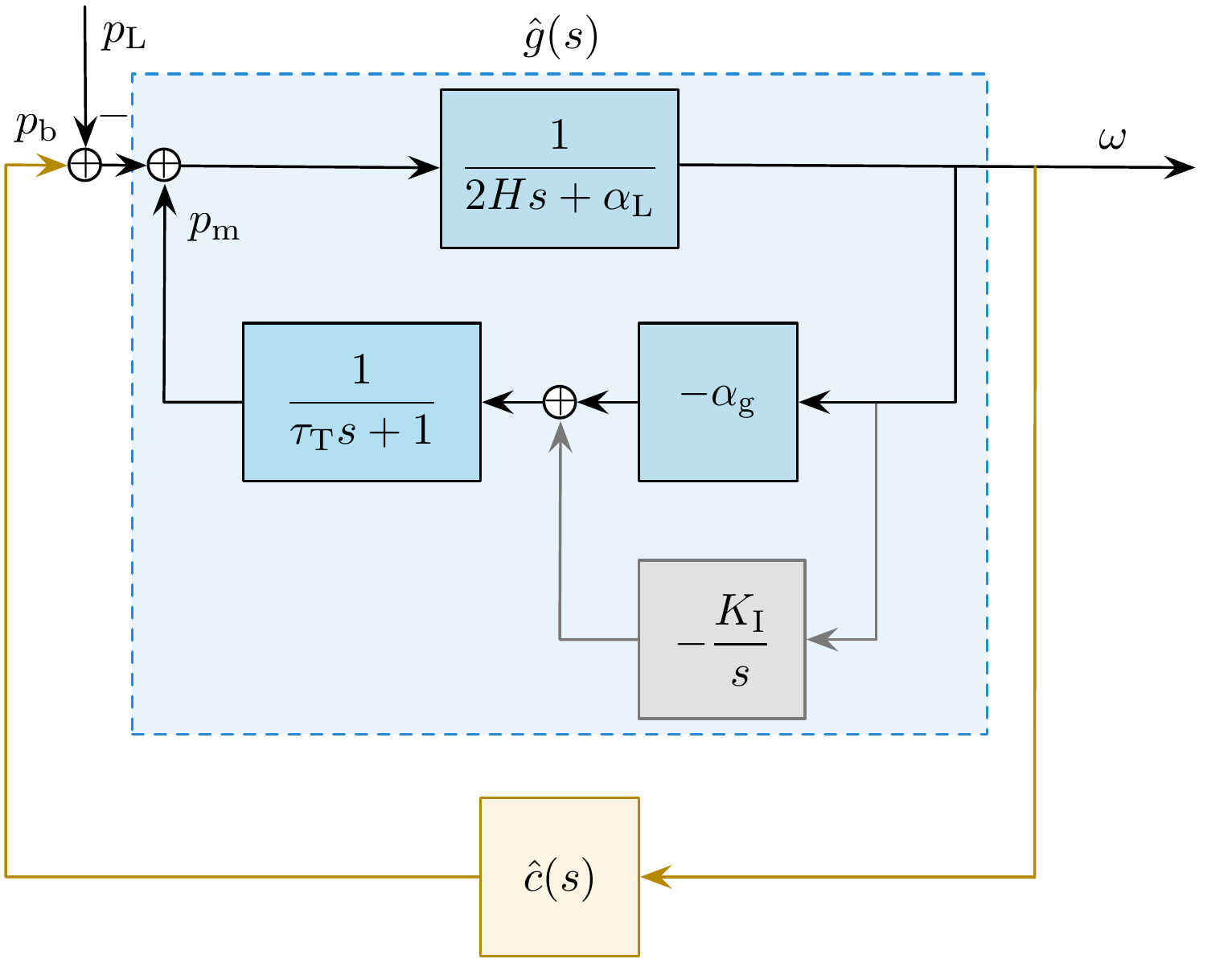}
\caption{Block diagram of an aggregated power system with frequency control from generators and storage units.}
\label{fig:1bus-model}
\end{figure}

For most of the manuscript, we will use the system parameters corresponding to the Great Britain system \cite{NG2016,NG2016standard}. We will use $P_\mathrm{B}=\SI{32}{\giga\VA}$ as a base power, and the value of the maximum power imbalance $\Delta P= \SI{1.8}{\giga\watt}$ corresponding to the loss of two biggest generation units, as specified in \cite{NG2016standard}. Under the high renewable penetration scenario, the total inertia of the system is expected to be around $\SI{70}{\giga\VA\second}$ which corresponds to $H=\SI{2.19}{\second}$ on the system base. The rest of the parameters are: $\tau_\mathrm{T}=\SI{1}{\second}$, $\alpha_\mathrm{g}=\SI{15}{pu}$, $\alpha_\mathrm{L}=\SI{1}{pu}$, and $K_\mathrm{I}=\SI{0.05}{pu\per\second}$.

\subsection{Performance Assessment of Frequency Control}

Since frequency deviation is volatile in a low-inertia power system, it is necessary to resort to certain measures to ensure frequency security, especially following major disturbances. Notably, for storage-based frequency control strategy design, not only control performance but also economic factors matter. Therefore, the performance metrics that are of our interest for comparing different control strategies are twofold: frequency response metrics and storage economics metrics. 

\subsubsection{Frequency Response Metrics}  The factors that are relevant to frequency security are:
\begin{itemize}
\item \emph{Steady-state frequency deviation} is the deviation of frequency from the nominal value after all the primary response is activated, i.e.,
\begin{equation}
     \Delta\omega:=\lim_{t\rightarrow\infty}\omega(t)\quad\text{with}\quad K_\mathrm{I}=0\,.
\end{equation}
The maximum allowed quasi-steady-state frequency deviation for the European and Great Britain systems is \SI{\pm200}{\milli\hertz} \cite{UCTLbook,NG2016standard}.
    \item \emph{Nadir} is the maximum frequency drop during a transient response, i.e., 
\begin{equation}\label{eq:Nadir}
	|\omega|_\infty := \max_{t\geq0} |\omega(t)|\,.
\end{equation}
For example, $|\omega|_\infty= \SI{800}{\milli\hertz}$ for the European system \cite{knap2015sizing} and $|\omega|_\infty= \SI{500}{\milli\hertz}$ for the Great Britain system \cite{NG2016,NG2016standard}. For microgrids, the maximum allowed frequency drop can be specified either by state standards or by some technical rules specific to the microgrid.    
    \item \emph{RoCoF} is the maximum rate of change of frequency, which usually occurs at the initial time instant, i.e.,
\begin{equation}\label{eq:rocof}
	|\dot{\omega}|_\infty := \max_{t\geq0}|\dot{\omega}(t)|\,.
\end{equation}
The highest RoCoF value allowed in the European system is $|\dot{\omega}|_\infty=\SI{0.5}{\hertz\per\second}$.

\end{itemize}

\subsubsection{Storage Economics Metrics} The two factors that significantly affect the cost of storage units are:
\begin{itemize}
\item \emph{Energy capacity} is the maximum amount of energy supply from storage during the whole transient duration, i.e.,
\begin{equation}
E_\mathrm{b,max} := \max_{t\geq0} E_\mathrm{b}(t)\,.    
\end{equation}
The maximum amount of energy supply directly determines the required storage capacity which, at present, represents the main contribution to the overall cost of storage systems.  
\item \emph{Maximum power} is the maximum amount of power output from storage during the whole transient duration, i.e.,
\begin{equation}
p_\mathrm{b,max} := \max_{t\geq0} p_\mathrm{b}(t)\,.    
\end{equation}
The maximum power output of the storage unit is also important since lower values of it mean that one can use inverters with lower installed capacity. 
\end{itemize}


\section{Conventional Control Strategy for Storage}
In this section we briefly analyze the most common traditional control strategy for storage-based frequency response  -- ``virtual inertia" (VI). We are mostly interested in its performance from the point of view of improving the Nadir and RoCoF, and we will assess the amount of power and energy required from storage to achieve certain performance level in the system. 

The most common VI strategy includes both inertial response (IR) and power-frequency response (PFR). It can be represented by the following effective storage transfer function $\hat{c}_\mathrm{vi}(s)$:
\begin{align}\label{vimain}
    \hat{c}_\mathrm{vi}(s):=-(m_\mathrm{v}s+\alpha_\mathrm{b})\;,\footnotemark[1]
\end{align}
where $m_\mathrm{v}$ is the IR constant and  $\alpha_\mathrm{b}$ is the PFR constant. We will use the subscript ``vi" to refer to this type of control strategy from storage. 
\footnotetext[1]{An additional low-pass filter is needed to make this transfer function proper. In all of our numerical models, we will use a low-pass filter with a cut-off frequency of \SI{5}{\hertz}. For simplicity we will omit it in formulas.}

For most of the analysis in the manuscript, when deriving analytic expressions, we will omit the secondary control since its purpose is to gradually drive the frequency back to nominal following a contingency and it does not significantly influence the transient frequency dynamics. Likewise, the load-frequency sensitivity coefficient $\alpha_\mathrm{L}$ will also be set to zero when deriving control laws -- this coefficient is typically of the order of unity for most power systems and its exact values are usually unknown. Setting $\alpha_\mathrm{L}$ to zero will make our results slightly conservative. We note though that all dynamic simulations will be done with secondary control and load-frequency sensitivity coefficient present.

\subsection{Coupled Nadir Elimination and RoCoF Tuning}

Under VI control strategy \eqref{vimain}, the frequency deviation $\omega_\mathrm{vi}$ following a step power imbalance in Laplace domain is given by
\begin{align}\label{eq:omegavi_main}
     \hat \omega_\mathrm{vi}(s) =\cfrac{-\Delta P\left(\tau_\mathrm{T} s + 1\right)}{s\left[\tilde{m}\tau_\mathrm{T}s^2+\left(\tilde{m}+\tau_\mathrm{T}\alpha_\mathrm{b}\right)s+\alpha_\mathrm{tot} \right]}\,,
\end{align}
where $\tilde{m}:=2H+m_\mathrm{v}$ is the compensated system inertia and $\alpha_\mathrm{tot} := \alpha_\mathrm{g}+\alpha_\mathrm{b}$ is the aggregate inverse droop of the system. 

By final value theorem, the steady-state frequency deviation is always determined by the aggregate inverse droop of the system as
\begin{align}\label{eq:fre-ss-vi}
\Delta\omega_{\mathrm{vi}} =\lim_{s\to0}s\hat \omega_\mathrm{vi}(s)=-\frac{\Delta P}{\alpha_\mathrm{tot}}\,.
\end{align}
With our present choice of values of parameters from the Great Britain system, the generator droop alone allows to fulfill the requirement on steady-state frequency deviation. This is easily seen by setting $\alpha_\mathrm{b}=0$ in \eqref{eq:fre-ss-vi}, which yields $\omega_{\infty,\mathrm{vi}}=\SI{-187.5}{\milli\hertz}$, a value within the range of \SI{\pm200}{\milli\hertz}. In contrast to the irrelevance of the steady-state frequency to the IR constant $m_{\mathrm{v}}$, the Nadir and RoCoF significantly depend on the choice of $m_{\mathrm{v}}$. As seen in Fig.~\ref{fig:fre-vi}, greater values of $m_{\mathrm{v}}$ will lead to decreases of both Nadir and RoCoF.
\begin{figure}[t!]
\centering
\includegraphics[width=0.72\columnwidth]{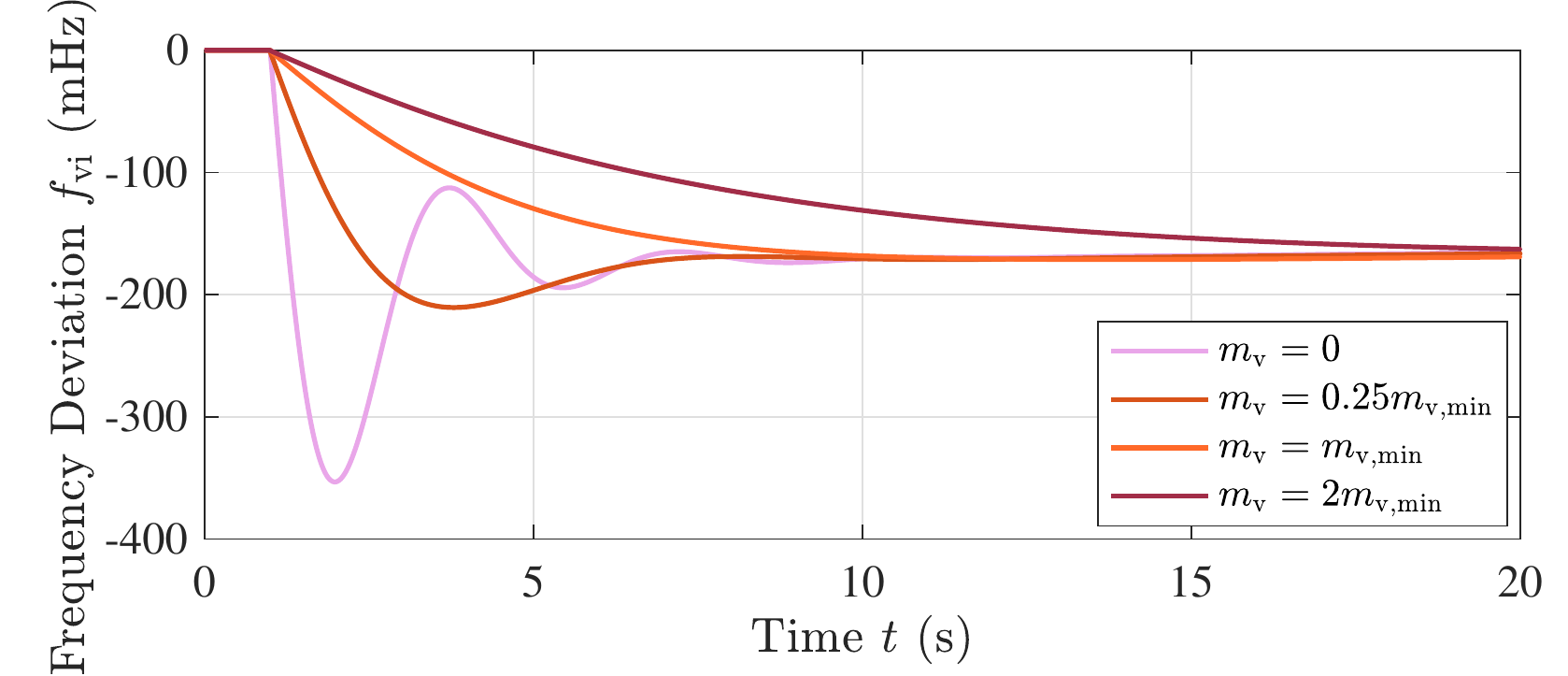}
\caption{Frequency deviations under virtual inertia control from storage with $\alpha_\mathrm{b}=0$ and different values of  $m_{\mathrm{v}}$.}
\label{fig:fre-vi}
\end{figure}

From Laplace domain expression \eqref{eq:omegavi_main}, analytic expression for $\omega_\mathrm{vi}(t)$  can be obtained following standard steps. The resulting expression is, however, rather cumbersome, thus we do not present it here explicitly. In order to find the frequency Nadir $|\omega|_{\infty,\mathrm{vi}}$ for arbitrary values of $m_\mathrm{v}$ and $\alpha_\mathrm{b}$, one needs to find the value of this function $\omega_\mathrm{vi}(t)$ at the time instant corresponding to its first minimum. This can be done by following the standard but unwieldy steps, the details of which can be found in \cite[Theorem  4]{jiang2019dynamic}.  Fig.~\ref{fig:nadirrocofvsIR} provides a plot of $|\omega|_{\infty,\mathrm{vi}}$ as a function of $m_{\mathrm{v}}$. It is evident that, for any given value of $\alpha_\mathrm{b}$, the Nadir gets shallower with the increase of the IR constant $m_{\mathrm{v}}$. Moreover, one can entirely remove Nadir by tuning $m_{\mathrm{v}}$ to be sufficiently large, where the critical value of $m_{\mathrm{v}}$ is determined by the following theorem.\footnotemark[2]\footnotetext[2]{In this manuscript, ``remove the Nadir" means ``remove the frequency response overshoot".}

\begin{figure}[t!]
\centering
\includegraphics[width=0.72\columnwidth]{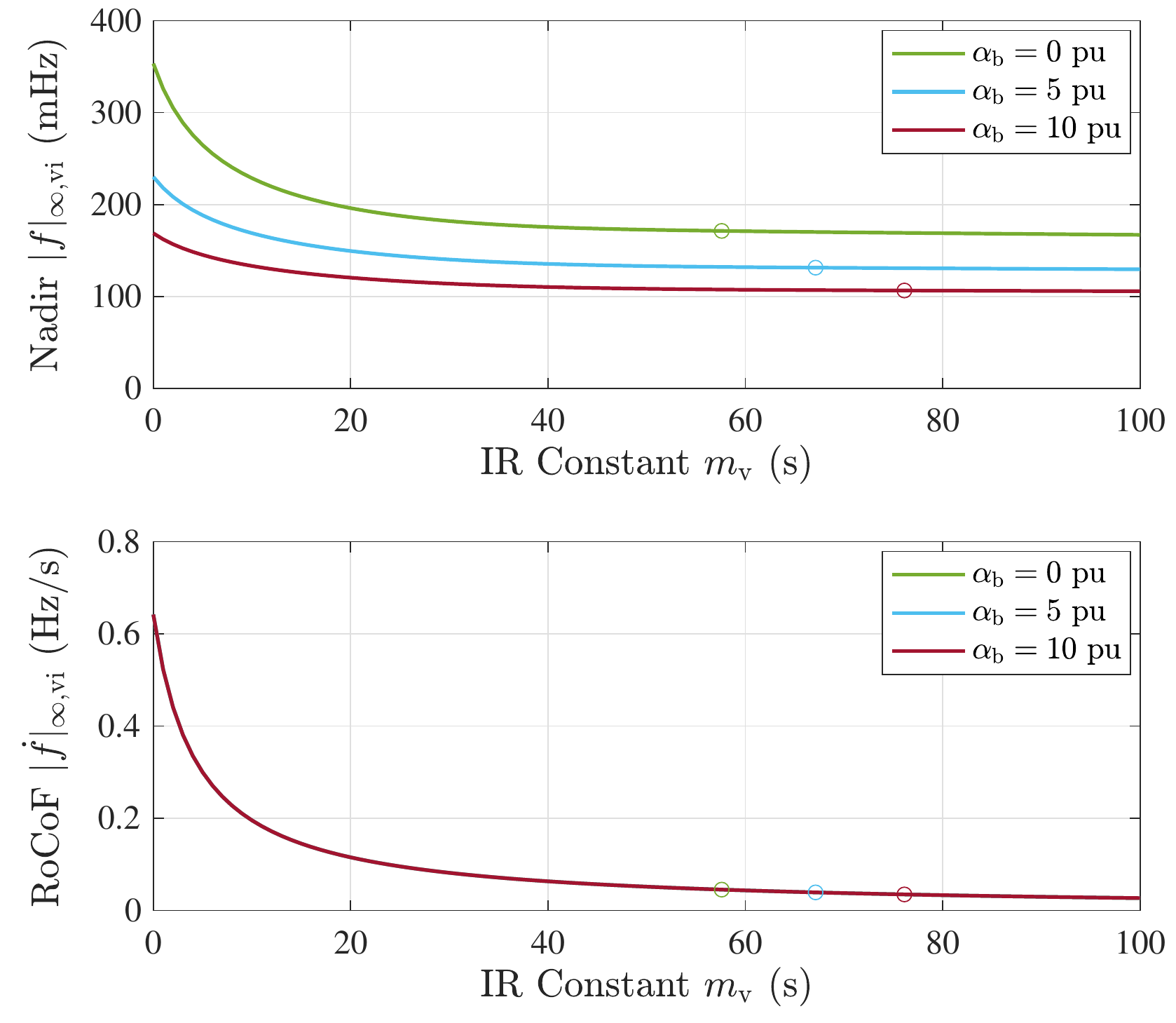}
\caption{Effect of the IR constant on Nadir and RoCoF for different chosen values of $\alpha_\mathrm{b}$, where circles denote the points corresponding to $m_\mathrm{v}=m_\mathrm{v,min}$.}
\label{fig:nadirrocofvsIR}
\end{figure}

\begin{thm}[Critical value of virtual inertia for removing Nadir] For a single-area power system described by \eqref{eq:dyn-onebus} and \eqref{eq:inverter-tf}, the step response under VI control, i.e., $\hat{c}(s)=\hat{c}_\mathrm{vi}(s)$, has no Nadir if 
\begin{align}\label{eq:mvmin}
   m_{\mathrm{v}} \geq m_\mathrm{v,min} :=\tau_\mathrm{T}\beta^2-2H\,,
\end{align}
with $\beta:=\sqrt{\alpha_\mathrm{g}}+\sqrt{\alpha_\mathrm{tot}}$.
\end{thm}
\begin{proof} Nadir occurs only if there exists some non-negative finite time instant $t_\mathrm{nadir}$ at which $\dot{\omega}_\mathrm{vi}(t_\mathrm{nadir})=0$. Therefore, a condition on $m_\mathrm{v}$ ensuring $\dot{\omega}_\mathrm{vi}(t)=0$ only when $t=\infty$ suffices to remove Nadir. Applying \cite[Theorem  4]{jiang2019dynamic} to \eqref{eq:omegavi_main}, we find that the VI control parameters $(\alpha_\mathrm{b}, m_{\mathrm{v}})$ should satisfy the following relations:
\begin{align}\label{eq:-nadirless-m-con}
\begin{cases}
    \tilde{m}^2-2\tau_\mathrm{T} \left(\alpha_\mathrm{b} +2\alpha_\mathrm{g}\right)\tilde{m}+\tau_\mathrm{T}^{2}\alpha_\mathrm{b}^2 \geq 0\\
    \tilde{m} - \tau_\mathrm{T} \alpha_\mathrm{b} \geq 0 
    \end{cases}\,.
\end{align}
The first quadratic inequality in $\tilde{m}$ above holds if $\tilde{m} \geq 
    \tau_\mathrm{T} \left(\sqrt{\alpha_\mathrm{g}}+\sqrt{\alpha_\mathrm{tot}}\right)^2$ or $\tilde{m} \leq
    \tau_\mathrm{T} \left(\sqrt{\alpha_\mathrm{g}}-\sqrt{\alpha_\mathrm{tot}}\right)^2$. 
However, only the former region satisfies the second condition in \eqref{eq:-nadirless-m-con}. This concludes the proof of the desired result.
\end{proof}

If it is recognized that $\alpha_\mathrm{b}$ is much smaller than $\alpha_\mathrm{g}$, an approximate expression for $m_\mathrm{v,min}$ in \eqref{eq:mvmin} can be obtained:
\begin{align}\label{eq:mvmin-linear}
   m_\mathrm{v,min} =2\tau_\mathrm{T} (2\alpha_\mathrm{g}+\alpha_\mathrm{b})- 2H \,.
\end{align}

The significance of eliminating the Nadir lies in the fact that the frequency security of the system can be certified by performing only simple \emph{algebraic} calculations so as to avoid running explicit \emph{dynamic} simulations. More precisely, given the expected maximum magnitude of power imbalance $\Delta P$ and the acceptable value for steady-state frequency deviation $\Delta\omega_{ \mathrm{d}}$, one can simply choose:
\begin{align}\label{eq:nonadir-tune}
\alpha_\mathrm{b}=&\ \max \left(\left|\frac{\Delta P}{\Delta\omega_{\mathrm{d}}}\right| - \alpha_\mathrm{g}, 0\right)\quad\text{and}\quad m_\mathrm{v}= m_\mathrm{v,min}\,.
\end{align}

Once both $m_\mathrm{v}$ and $\alpha_\mathrm{b}$ are determined, one can calculate RoCoF as:
\begin{align}
    |\dot{\omega}|_{\infty,\mathrm{vi}} =& \left|\lim_{s\to\infty} s^2\hat \omega_\mathrm{vi}(s)\right| =  \frac{\Delta P}{\tilde{m}}\,,\label{eq:rocof-vi}
\end{align}
which is inversely proportional to the compensated system inertia. Thus, under VI control, there is a coupling between Nadir elimination and RoCoF tuning. If one adopts the tuning given in \eqref{eq:nonadir-tune}, then the RoCoF is fixed to be $\Delta P /(2H+m_\mathrm{v,min})$. Yet, the value of $m_\mathrm{v,min}$ required usually has rather high values, which leads to a too small RoCoF and thus a very long settling time, as shown in Fig.~\ref{fig:fre-vi}. If one hopes to adjust the RoCoF appropriately so as to let the frequency evolves with a moderate rate, then the Nadir cannot be removed. We also note that, according to \eqref{eq:mvmin-linear}, $m_\mathrm{v,min}$ is very sensitive to turbine time-constant, and the values shown in Fig.~\ref{fig:nadirrocofvsIR} correspond to a rather modest value of $\tau_\mathrm{T}=\SI{1}{\second}$. For slower governors and turbines, the requirements on $m_\mathrm{v,min}$ will be even higher.

\subsection{Power and Energy Requirements on Storage }\label{sec:VI_power}
In order to quantify the required amount of power rating of the storage unit for a given control strategy $\hat{c}_\mathrm{vi}(s)$, one needs to find the maximum of $p_\mathrm{b}(t)$ during the whole transient. The procedure is rather straightforward, but the explicit expression for $p_\mathrm{b}(t)$ in time domain is very cumbersome for arbitrary values of $m_\mathrm{v}$ and $\alpha_\mathrm{b}$. The top panel of Fig.~\ref{fig:saturation-vi} shows the maximum storage power as a function of IR constant $m_\mathrm{v}$ for different values of PFR constant $\alpha_\mathrm{b}$ (for power/energy requirement figures we use per unit system based on a disturbance size for simpler comparison). Comparing with the top panel in Fig.~\ref{fig:nadirrocofvsIR}, we observe that, for $m_\mathrm{v}$ less than $m_\mathrm{v,min}$, the Nadir and power rating are quite sensitive to variations in $m_\mathrm{v}$, yet, for $m_\mathrm{v}$ greater than $m_\mathrm{v,min}$, they are practically insusceptible to changes in $m_\mathrm{v}$. This implies that $m_\mathrm{v}=m_\mathrm{v,min}$ plays the role of a saturation point after which an increase in power rating of storage system does not provide any benefit to a decrease in Nadir. We note, though, that the actual values of $m_\mathrm{v,min}$ correspond to huge additional inertia -- about ten times more than the system natural inertia, and. 

It is evident from Fig.~\ref{fig:saturation-vi} that, for the special case of $m_\mathrm{v}=m_\mathrm{v,min}$, the maximum power $p_\mathrm{b,max}$ is almost independent of the PFR constant $\alpha_\mathrm{b}$, so we can set it to zero to simplify the expression for $p_\mathrm{b}(t)$. Thus, for the case $m_\mathrm{v}=m_\mathrm{v,min}$ and $\alpha_\mathrm{b}=0$, one has the following expression:  
\begin{align}\label{eq:p_of_t_vi}
p_\mathrm{b,vi}(t)=\Delta P\left(1-\frac{H}{2\tau_\mathrm{T}\alpha_\mathrm{g}}\right)\left(1+\frac{t}{2\tau_\mathrm{T}}\right) e^{-\frac{t}{2\tau_\mathrm{T}}}\,,
\end{align}
for which the maximum value occurs at $t=0$ and is equal to $p_\mathrm{b,max}=\Delta P \left[1-H/(2\tau_\mathrm{T}\alpha_\mathrm{g})\right]$ and is very close to the disturbance size $\Delta P$ for realistic values of parameters.

\begin{figure}[t!]
\centering
\includegraphics[width=0.72\columnwidth]{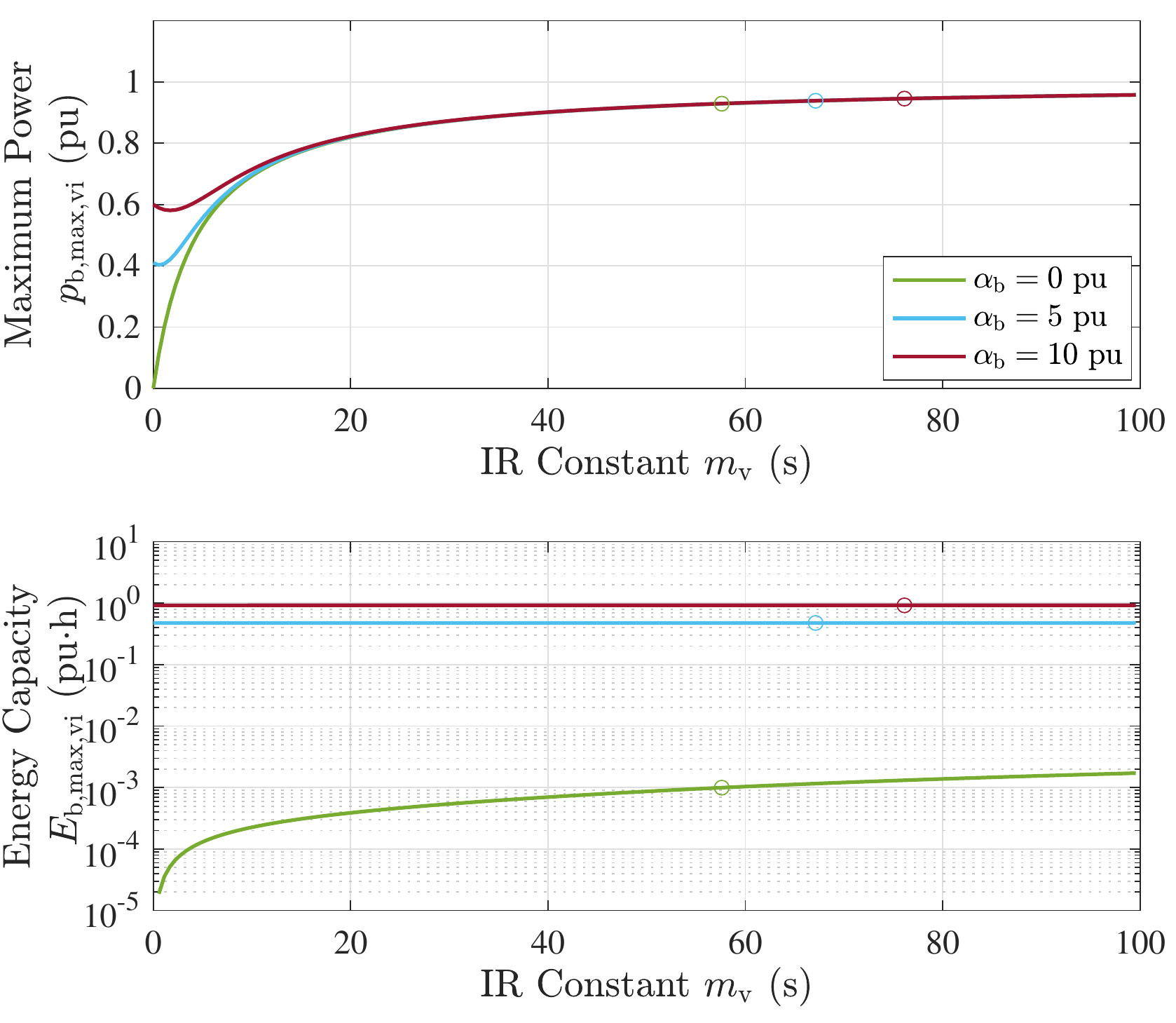}
\caption{Effect of the IR constant on maximum power and energy capacity requirements, where circles denote the points corresponding to $m_\mathrm{v}=m_\mathrm{v,min}$.}
\label{fig:saturation-vi}
\end{figure}

The storage energy capacity required to execute $\hat{c}_\mathrm{vi}(s)$ is predominantly determined by the values of $\alpha_\mathrm{b}$ and the system secondary control constant $K_\mathrm{I}$ (the bottom panel in Fig.~\ref{fig:saturation-vi}). For zero values of $\alpha_\mathrm{b}$, the energy capacity required will be very small ($10^{-3}$ pu$\cdot$h or less). Note that, in practice, unless $\alpha_\mathrm{b}$ is substantial, the minimum storage energy capacity will be determined by the $C$-rate of the batteries used, hence, the exact value of $E_{\mathrm{b,max,vi}}$ is of little importance. For values of $\alpha_\mathrm{b}$ that are not very small, an approximate formula $E_{\mathrm{b,max,vi}} \approx \alpha_\mathrm{b}/K_\mathrm{I}$ can be used. This suggests that higher secondary control gains tend to reduce the required storage energy capacity. 

To summarize, VI can be an effective tool to improve transient frequency performance, however, RoCoF and Nadir become coupled under this type of control. Modest amounts of IR from storage can somewhat improve the frequency Nadir, however, complete elimination of Nadir will require massive amounts of IR, which makes it impractical.

\section{Frequency Shaping Control}

Although VI seems to be a straightforward choice, inverter-interfaced storage units are potentially capable of executing a much wider class of control strategies. In this section we propose a novel control strategy -- \emph{frequency shaping control} -- which is able to decouple the Nadir elimination task from the RoCoF tuning one. The general idea behind the frequency shaping control is to effectively turn the system frequency dynamics into a first-order one, which is dependent on two control parameters, by employing a special form of storage response function $\hat{c}(s)$. Such a first-order response has no Nadir naturally, while tuning of the two parameters will provide the ability to adjust both the steady-state frequency deviation and the RoCoF, independently.

\subsection{Controller Design}

For designing the needed frequency shaping control (which we will denote as $\hat{c}_\mathrm{fs}(s)$), let us consider a second-order transfer function of the following form:\footnotemark[3]

\begin{equation}\label{eq:c_fs}
    \hat{c}_\mathrm{fs}(s):=-\frac{A_1 s^2 +A_2s + A_3}{\tau_\mathrm{T} s + 1}\,,
\end{equation}
where $\tau_\mathrm{T}$ is the system turbine time constant from \eqref{eq:tur-dyan}, $A_1$, $A_2$, and $A_3$ are tunable control parameters. Then, the desired frequency shaping control is determined by the next \emph{theorem}.  
\footnotetext[3]{As in virtual inertia control, an additional low-pass filter is also needed in frequency shaping control for the same reason.}

\begin{thm}[Frequency shaping]\label{theorem:shaping}
The single-area system from Fig.~\ref{fig:1bus-model} will respond to a power imbalance $-p_\mathrm{L}$ as a first-order system of the form:
\begin{equation}\label{eq:hd}
    \hat{h}(s) = \frac{1}{as+b}\qquad\text{i.e.}\qquad\hat{\omega}(s)=-\hat{h}(s)\hat{p}_\mathrm{L}(s)
\end{equation}
with $a$ and $b$ being positive constants, if the corresponding storage frequency response function $\hat{c}(s)$ is given by \eqref{eq:c_fs} with the following values of constants: 
\begin{subequations}\label{eq:A_parameters}
\begin{align}
    A_1 =&\ \tau_\mathrm{T}\left(a-2H\right)\,,\\
    A_2 =&\ b\tau_\mathrm{T}+ a - 2H\,,\\
    A_3 =&\ b-\alpha_\mathrm{g}\,.
\end{align}
\end{subequations}
In this case, the system frequency will experience no Nadir and the steady-state frequency deviation {$\Delta\omega$} and the  RoCoF {$|\dot{\omega}|_{\infty}$} will be determined by the following expressions:  
\begin{align}\label{eq:ab-tune}
    {\Delta\omega}= -\frac{\Delta P}{b} \qquad\text{and}\qquad {|\dot{\omega}|_{\infty} }= \frac{\Delta P}{a} 
\end{align}
when $\hat{p}_\mathrm{L}(s)=\Delta P/s$.
\end{thm}
\begin{proof}
Let the desired closed-loop transfer function from $-p_\mathrm{L}$ to $\omega$ be a first-order one given by \eqref{eq:hd}. Then, using explicit expression for the generator/turbine transfer function $\hat{g}(s)$ and \eqref{eq:hd}, one can directly solve  for the desired storage control strategy as
\begin{equation}
    \frac{\hat{h}(s)-\hat{g}(s)}{\hat{h}(s)\hat{g}(s)}=-\frac{A_1 s^2 +A_2s + A_3}{\tau_\mathrm{T} s + 1}=:\hat{c}_\mathrm{fs}(s)
\end{equation}
with $A_1$, $A_2$, and $A_3$ given by \eqref{eq:A_parameters}.

Next, applying initial and final value theorems to \eqref{eq:hd}, we find that $a$ and $b$ satisfy the following relations:
\begin{subequations}
\begin{align}
    {|\dot{\omega}|_{\infty} }=&\ \Delta P\lim_{s\to\infty} s^2\frac{\hat{h}(s)}{s}
    =\frac{\Delta P}{a}\,,\label{eq:rocof-a}\\
    \Delta\omega=& -\Delta P\lim_{s\to0}s\frac{\hat{h}(s)}{s}
    =-\frac{\Delta P}{b}\,,\label{eq:fss-b}
\end{align}
\end{subequations}
which are identical to \eqref{eq:ab-tune}. 
\end{proof}

Theorem \ref{theorem:shaping} allows one to tune the storage frequency response strategy that guarantees Nadir-less response for the whole system while also providing the pre-set values for RoCoF and steady-state frequency deviation. 
However, such a tuning can lead to sub-optimal use of the storage capabilities if the system response without storage already provides satisfactory performance in terms of either RoCoF or steady-state frequency deviation (or both). Therefore, the actual tuning will depend on the existing system performance. Suppose the desired values of the RoCoF and steady-state frequency deviation are $|\dot{\omega}|_{\infty \mathrm{d}}$ and $\Delta\omega_{\mathrm{d}}$, respectively. Overall, four cases are possible: 

\subsubsection{Case 1} System's response suffices to provide satisfactory performance for both RoCoF and steady-state frequency deviation. In this case, one needs to use their actual values instead of the maximum allowed ones for tuning $a$ and $b$. Thus, the optimal settings are:   
\begin{equation}
a=2H\qquad\text{and}\qquad b= \alpha_\mathrm{g}\,.
\end{equation}
Here, the effect of storage is to eliminate frequency Nadir while keeping RoCoF and steady-state frequency deviation unchanged. 
\subsubsection{Case 2} System's response suffices to provide satisfactory performance for RoCoF but not for steady-state frequency deviation -- there is enough inertia but not enough primary response from generators. Then, $a$ and $b$ should be:
\begin{equation}
a=2H\qquad\text{and}\qquad b= -\frac{\Delta P}{\Delta\omega_{\mathrm{d}}}.
\end{equation}

Note that the above two cases correspond to the so-called iDroop -- a dynamic droop tuning reported by us recently \cite{jiang2019dynamic}: 
\begin{equation}\label{eq:co-idroop-nonadir}
\hat{c}^\star_\mathrm{iDroop}(s) =\frac{ \alpha_\mathrm{g} }{\tau_\mathrm{T} s+1}-\left(\alpha_\mathrm{g} + \alpha_\mathrm{b}\right)\,.
\end{equation}
This type of control is capable of eliminating Nadir and at the same time improve the steady-state frequency deviation (if needed), but does not affect the RoCoF.

\subsubsection{Case 3} System's response suffices to provide satisfactory steady-state frequency deviation but not RoCoF -- there is enough primary response but not enough inertia. In this case:
\begin{equation}
a=\frac{\Delta P }{|\dot{\omega}|_{\infty\mathrm{d}}}\qquad\text{and}\qquad b= \alpha_\mathrm{g}.
\end{equation}

\subsubsection{Case 4} System's response is insufficient to provide satisfactory steady-state frequency deviation and RoCoF -- there is lack of both primary response and inertia. In this case:
\begin{equation}
a=\frac{\Delta P}{|\dot{\omega}|_{\infty\mathrm{d}}}\qquad\text{and}\qquad b= - \frac{\Delta P}{\Delta\omega_{\mathrm{d}}}.
\end{equation}

For all four cases, the frequency shaping control, as its name suggests, makes the system frequency response effectively first-order, thus eliminating Nadir. Moreover, it also ensures that both RoCoF and steady-state frequency deviation are within the pre-specified limits $|\dot{\omega}|_{\infty\mathrm{d}}$ and  $\Delta\omega_{\mathrm{d}}$, respectively. Fig.~\ref{fig:dyn-fre-ab0_nodb} illustrates the well-shaped frequency response under two different tunings of  frequency shaping control (corresponding to Cases $1$ and $3$) compared with VI control (with $m_\mathrm{v}=m_\mathrm{v,min}$) and  no storage base scenario.  

To explicitly demonstrate the difference between frequency shaping control and VI, Fig.~\ref{fig:tradeoff-fs} shows Nadir as a function of the RoCoF. It is obvious that for VI those two metrics are coupled, while the frequency shaping control provides us the freedom to tune RoCoF without sacrificing Nadir elimination.

\begin{figure}[t!]
\centering
\includegraphics[width=0.72\columnwidth]{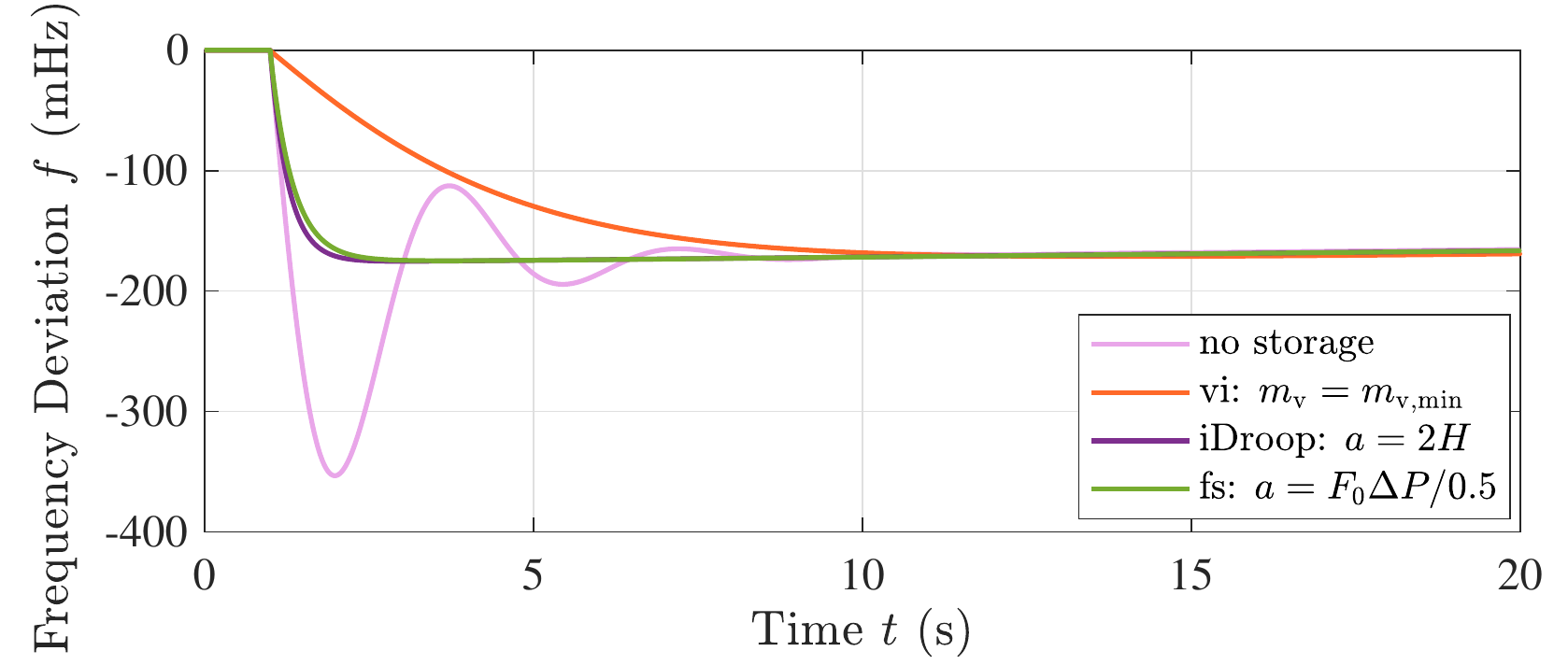}
\caption{Comparison of frequency deviations with Nadir eliminated under virtual inertia and frequency shaping control to a step
power imbalance for $\alpha_\mathrm{b} = 0$.}
\label{fig:dyn-fre-ab0_nodb}
\end{figure}

\begin{figure}[t!]
\centering
\includegraphics[width=0.72\columnwidth]{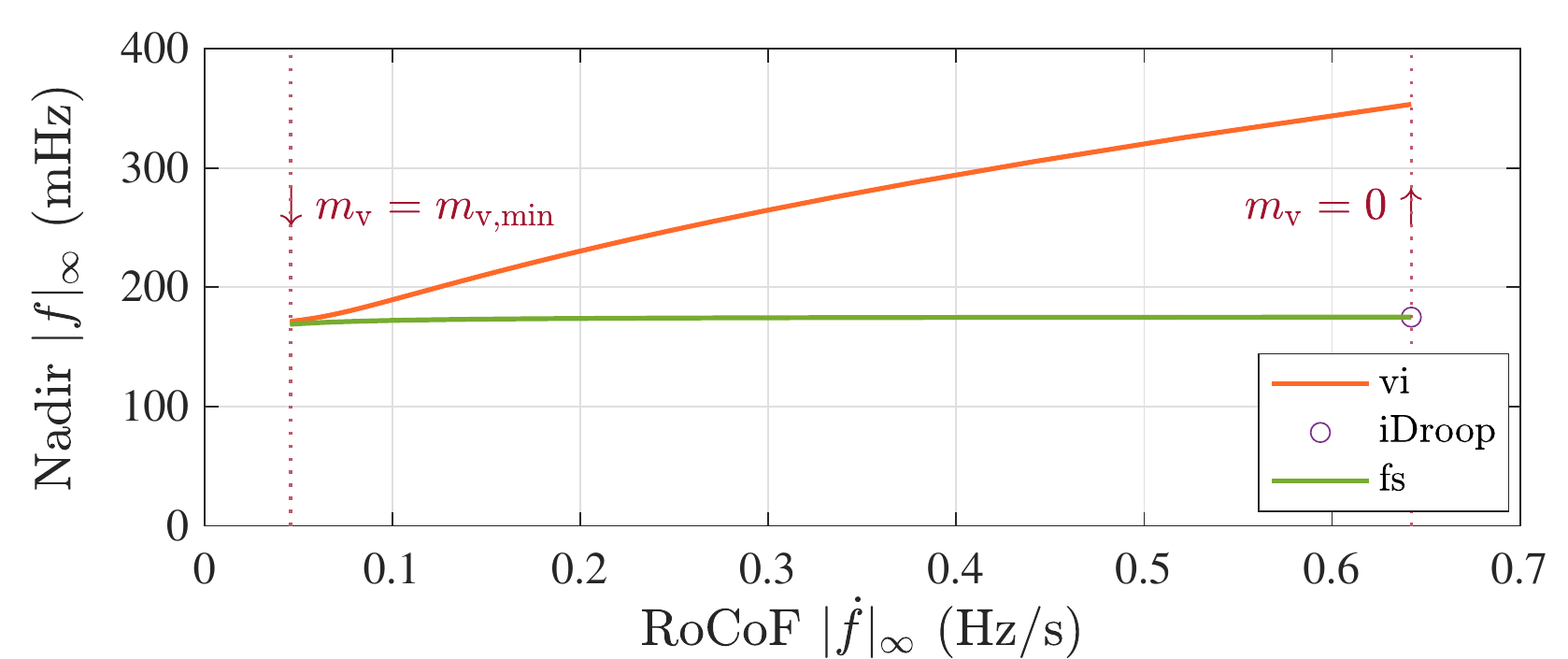}
\caption{Nadir as a function of RoCoF under frequency shaping control for $\alpha_\mathrm{b}=0$ and $a$ within the range of $\left[2H,(2H+m_\mathrm{v,min}) \right]$.}
\label{fig:tradeoff-fs}
\end{figure}

\begin{figure}[t!]
\centering
\includegraphics[width=0.72\columnwidth]{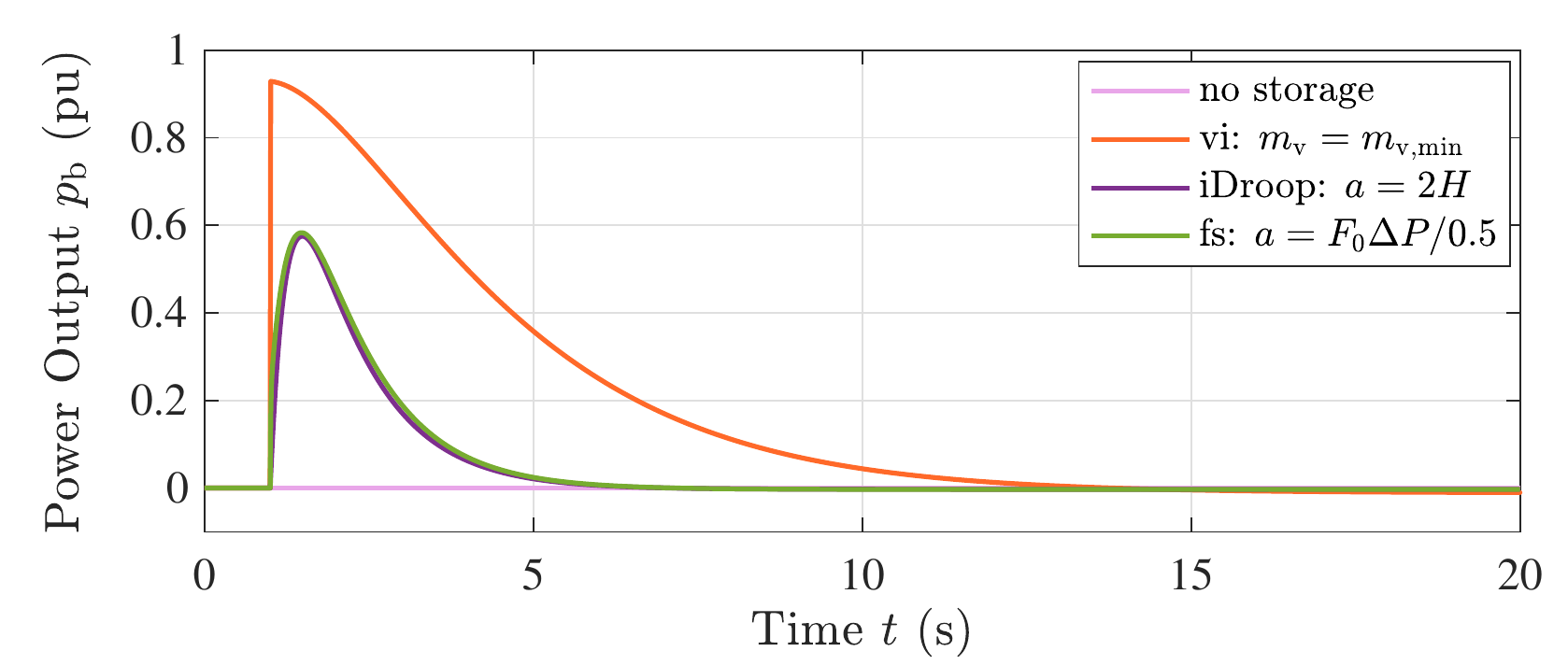}
\caption{Comparison of storage power transient responses with Nadir eliminated under virtual inertia and frequency shaping control to a step
power imbalance for $\alpha_\mathrm{b} = 0$.}
\label{fig:dyn-storage-ab0_nodb2}
\end{figure}

\subsection{Power and Energy Requirements on Storage}
We next quantify the required amount of storage power to execute the frequency shaping control. Provided that the control is given by \eqref{eq:c_fs}, the storage power output (following a step power imbalance) in Laplace domain is: 
\begin{align}
\hat{p}_\mathrm{b,fs}(s)
\!=\! -\frac{\Delta P}{s}\hat{h}(s)\hat{c}_\mathrm{fs}(s)\!=\!\frac{\Delta P\!\left(A_1 s^2\!+\! A_2 s \!+\! A_3\right)}{s\left(as+b\right)\left(\tau_\mathrm{T}s+1\right)}\,.    
\end{align}
From here, an explicit (and rather cumbersome) expression for power in time domain can be found. Fig. \ref{fig:dyn-storage-ab0_nodb2} shows the storage power output as a function of time for the four control strategies. Clearly, the proposed frequency shaping control outperforms the VI control -- it requires up to $40\%$ less storage power. In addition, the duration of the power peak is much shorter for frequency shaping control, which can allow to decrease the storage installed power even more.

The energy capacity requirement for the frequency shaping control is mostly determined by the effective battery droop, similarly to VI strategy. For Cases $1$ and $3$ of the previous subsection, the storage effective droop is zero, so the energy requirement is very small and, like for VI, capacity will be mostly determined by the $C$-rate of batteries used.\footnotemark[4] For Cases $2$ and $4$, the storage is supposed to participate in the steady-state frequency response and the energy capacity requirement will be significantly higher. Similarly to VI control, it can be estimated as $E_{\mathrm{b,max,fs}} \approx \alpha_{\mathrm{b}}/K_\mathrm{I}$, where $\alpha_\mathrm{b} \equiv A_3 = b - \alpha_\mathrm{g}$ -- the storage effective PFR constant.  

\footnotetext[4]{Energy requirements for frequency shaping control will be formally less than that for virtual inertia, but this is irrelevant in practice due to $C$-rate limitations.}

The intuition behind the effectiveness of the frequency shaping control is that it is able to take the most advantage of the system natural frequency response capabilities. While VI can provide performance increase for both RoCoF and Nadir, there is no way to decouple them in order to optimize the control effort. Frequency shaping control, on the contrary, provides virtual inertia only if it is needed to secure acceptable RoCoF value, and only with the minimum value needed. Frequency Nadir is then taken care of by a different contribution -- iDroop, that is able to guarantee the effective first-order system dynamics. Thus, frequency shaping control leverages the knowledge of the system inertia, primary response, and turbine dynamics in order to provide a more optimal response by making full use of the system's own control capabilities.

\section{Generalization for Multi-machine and Multiple-area Systems}\label{sec:multiple}

Storage control strategy described by \eqref{eq:c_fs} was derived for a single-machine representation with a simplified model for generator turbine. Actually, the same methodology can be applied to derive control strategy for more general cases. In this section we provide a generalization of the method for multi-machine and multiple-area systems with arbitrary models for governors and turbines. 

We start from deriving the closed-loop power-frequency response for a multi-machine single-area system. Let $H_i$ be the inertia constant of $i$'s machine (for all the variables we denote machines by a lower index $i$) and $\hat{T}_i(s)$ be a combined transfer function of its governor and turbine, i.e., $\hat{p}_{\mathrm{m},i}(s)=-\alpha_{\mathrm{g},i} \hat{T}_i(s) \hat{\omega}_i(s)$. Note that $T_i(0)=1$ for every machine. Now, the multi-machine closed-loop frequency response to power imbalance is:
\begin{equation}
    \hat{g}(s) = \frac{1}{\sum_i (2 H_i s + \alpha_{\mathrm{g},i}\hat{T}_i(s))}\,.
\end{equation}
Similarly to how we did before, we state that the additional storage frequency control strategy should transform the overall system response to an effective first-order form given by \eqref{eq:hd} with the constants $a$ and $b$ determining the system RoCoF and steady-state frequency deviation respectively. In the case of multi-machine system, the storage frequency response can be provided either in aggregated or fully decentralized way. In the latter case, one can think that each machine is ``matched" by a corresponding storage response function $\hat{c}_i(s)$ of individual storage units in such a way that the overall system dynamics satisfies \eqref{eq:hd}. Then, the following relation should be satisfied:        
\begin{equation}
    \sum_i(2H_i s + \alpha_{\mathrm{g},i}\hat{T}_i(s) - \hat{c}_i(s))= as+b\,.
\end{equation}
Let us now represent the response functions $\hat{c}_i(s)$ of individual storage units in the following way:
\begin{equation}\label{eq:ind_storage}
    \hat{c}_i(s) = -(m_i s - \alpha_{\mathrm{g},i}\hat{T}_i(s) +\alpha_{\mathrm{g},i} + \alpha_{\mathrm{b},i})\,,  
\end{equation}
where the first term represents the virtual inertia response that is responsible for RoCoF, and the other terms represent the dynamic droop that is responsible for Nadir and steady-state frequency deviation. 

Derivation of required values for $m_i$ and $\alpha_{\mathrm{b},i}$ is somewhat similar to the derivation for a single-machine system. First of all, if the system's natural response is sufficient to provide satisfactory RoCoF and steady-state frequency response, then all $m_i$ and $\alpha_{\mathrm{b},i}$ can be set to zero, so that control strategy for the storage units becomes:
\begin{equation}
    \hat{c}_i(s) = \alpha_{\mathrm{g},i}\hat{T}_i(s) -\alpha_{\mathrm{g},i}  \,, 
\end{equation}
which is a direct generalization of \eqref{eq:co-idroop-nonadir}. 
In the case either RoCoF or steady-state response (or both) need to be improved by the storage, required storage $m_i$ and $\alpha_{\mathrm{b},i}$ can be determined from the following relations (we assume minimum required settings):
\begin{subequations}\label{eq:multi_m_alpha}
\begin{align}
    \sum_i m_i = &\ \frac{\Delta P}{ |\dot{\omega}|_{\infty\mathrm{d}}} - 2\sum_i H_i, \label{eq:multi_m} \\
    \sum_i \alpha_{\mathrm{b},i} = &\ -\frac{\Delta P}{ \Delta\omega_{\mathrm{d}}} - \sum_i \alpha_{\mathrm{g},i}\,. \label{eq:multi_alpha}
\end{align}
\end{subequations}
Here, $|\dot{\omega}|_{\infty\mathrm{d}}$ and $\Delta\omega_{\mathrm{d}}$ are the maximum allowed values of RoCoF and steady-state frequency deviation respectively (the latter will also correspond to frequency ``Nadir" for first-order response). From the mathematical point of view, as long as \eqref{eq:multi_m_alpha} are satisfied, the assignment of individual values $m_i$ and $\alpha_{\mathrm{b},i}$ can be done arbitrarily. From the practical point of view, contribution according to generator installed power might make sense. Another way, which could be more reasonable, is to set certain minimum requirements for generator inertia and droop gain, and then storage units are tuned to provide some additional $m_i$ and/or $\alpha_{\mathrm{b},i}$ only for those generators that do not meet the threshold with their conventional capabilities. 

Another important practical aspect is the tuning of storage units to provide the response $\hat{T}_i(s)$ that matches the corresponding governor-turbine dynamics. It is possible to tune the storage using the fully detailed governor model. However, even a simple second-order reduced model obtained from $\hat{T}_i(s)$ by balanced truncation procedure provides remarkably good performance. Fig.~\ref{fig:1bus3machine-fre} shows frequency dynamics of a three-machine equivalent system. Two of the machines are equipped with steam turbines modelled using IEEEG1 governor model (which is a good representation for real-life systems \cite{chavez2014governor}), and one of the machines is equipped with hydro turbine. Two of the machines are supposed to have sufficient inertia, while the third one needs to implement additional virtual inertia from the storage unit. The storage frequency controls for all three machines are designed using truncated second-order governor models. Obviously, the frequency response of the whole system under frequency shaping control is very close to the desired first-order one, which has no Nadir and satisfies the constraints on RoCoF and steady-state frequency deviation.    

\begin{figure}[t!]
\centering
\includegraphics[width=0.72\columnwidth]{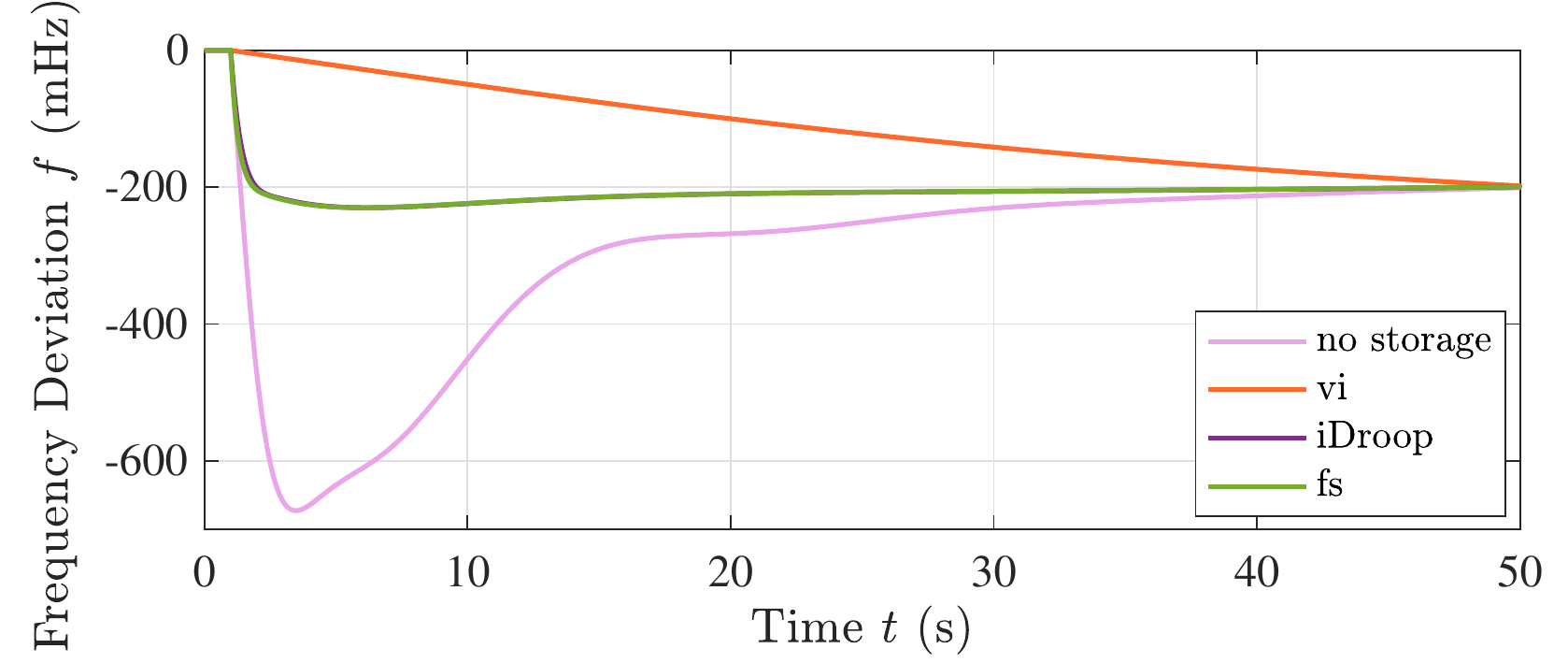}
\caption{Comparison of frequency deviations under different controllers in a single-area three-machine system when a step power imbalance is introduced.}
\label{fig:1bus3machine-fre}
\end{figure}

Finally, we provide the generalization of the method to multiple-area systems. We note that in this case it is impossible to get the ideal first-order response from every area, i.e., it is impossible to eliminate frequency Nadir strictly. However, it is possible to significantly limit both RoCoF and Nadir for each area. We also note that in many cases even when the system is considered to be multiple-area from the point of view of secondary control (i.e., consisting of a number of Balancing Authorities), its transient frequency dynamics can still be described with a good accuracy using a single-area approximation. In other words, the frequency dynamics of every machine in the system is rather close to the dynamics of its COI, which allows one to still use relations \eqref{eq:ind_storage} and \eqref{eq:multi_m_alpha} to tune the storage response. In the case of true multiple-area system (i.e., when there are long interconnection lines present), the following simple generalizations should be made. 

Each area $\alpha$ provides its maximum instantaneous power imbalance $\Delta P^{\alpha}$, then the IR constants $m_i^{\alpha}$ are tuned according to \eqref{eq:multi_m} separately for each area (i.e., as if the areas are isolated). Likewise, the maximum allowed steady-state frequency deviation $\Delta\omega_{\mathrm{d}}^\alpha$ for each area $\alpha$ will serve as a conservative \emph{lower boundary} for the frequency Nadir in the corresponding area, so that the minimum values for PFR constants $\alpha_{\mathrm{b},i}^\alpha$ can be set for each area by \eqref{eq:multi_alpha}. Finally, applying \eqref{eq:multi_alpha} again for the whole system, with $\Delta P$ in the right-hand side being the largest among all the $\Delta P^{\alpha}$ for individual areas, we get the maximum value of the system steady-state frequency deviation. Then, additional tuning of the PFR constants can be made, if the steady-state frequency deviation needs to be further improved. Fig. \ref{fig:2bus-fre} provides an illustration of our control performance for a two-area system. It shows the frequencies in both areas and the frequency of COI, following a contingency in area $2$. Horizontal dashed lines show the steady-state frequency deviations for separate areas -- we see that frequencies in both areas are always well above these lines. 


\section{Conclusion}

We have presented a new type of frequency control strategy for energy storage units, which allows to completely eliminate frequency Nadir by making the system dynamics effectively first-order. Our control method significantly outperforms the conventional VI strategy, requiring up to $40\%$ less peak power from storage, while also significantly reducing the duration of the peak-power response. The effectiveness of our strategy is based on its ability to utilize the system frequency response capability effectively withdrawing the storage response as the generator turbine increases its power output.    
      
Nadir-less dynamics can allow to completely revise the security assessment procedures, which can now be done using simple \emph{algebraic} calculations, rather than dynamic simulations. Moreover, we envision that the ``shaping" of generators' power-frequency response by storage units can provide benefits beyond the frequency control itself. Among the straightforward applications include mitigation of turbine effort for frequency control, and small-signal and transient stability enhancements. Another direction is the development of control loops for power electronics in order to provide device-level execution of the proposed frequency shaping control.   

\begin{figure}[t!]
\centering
\includegraphics[width=0.72\columnwidth]{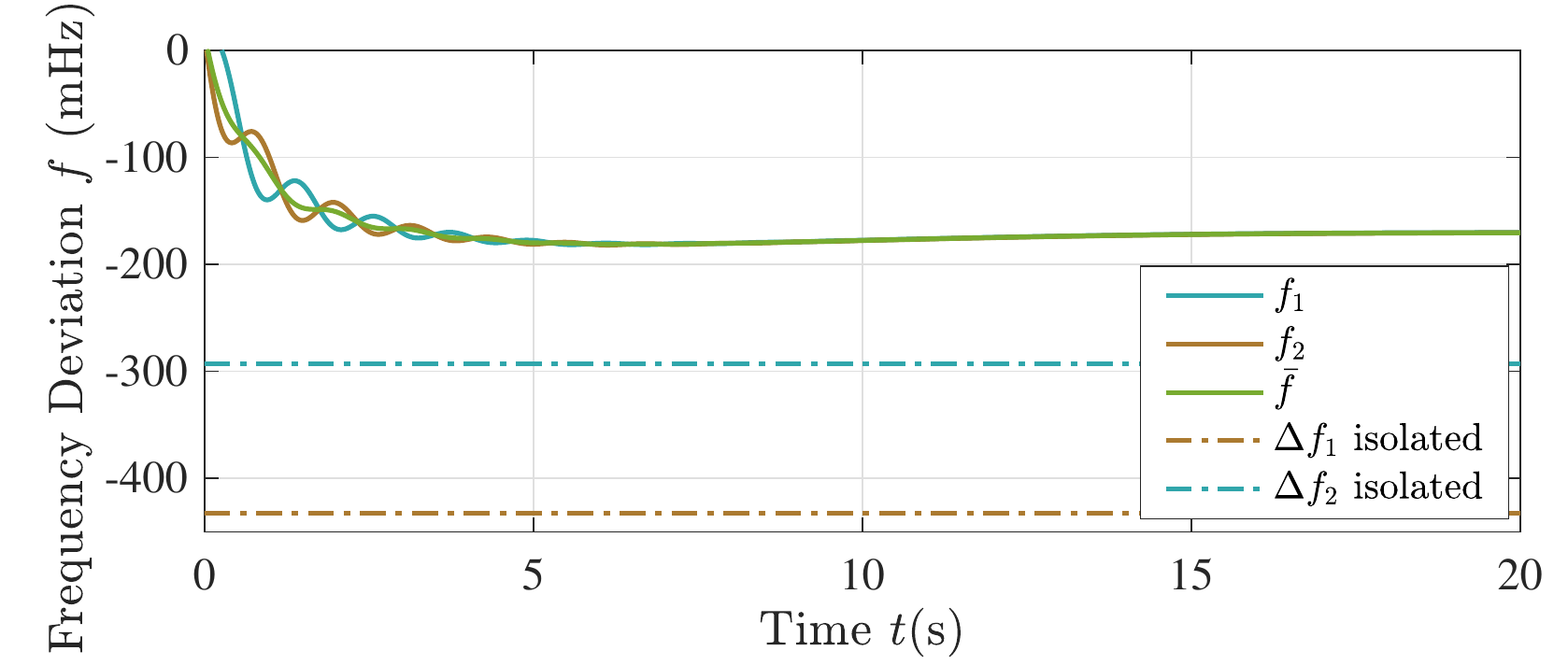}
\caption{Two-area system frequency responses to a step power imbalance in area $2$. Both areas are under frequency shaping control. }
\label{fig:2bus-fre}
\end{figure}

\ifCLASSOPTIONcaptionsoff
  \newpage
\fi


\bibliographystyle{IEEEtran}
\bibliography{TPS}

%









\end{document}

%% file: edits.tex
\usepackage{ifthen}
\newboolean{showcomments}
\setboolean{showcomments}{true}
\usepackage[usenames,dvipsnames]{color}
\usepackage{todonotes}

\definecolor{bleudefrance}{rgb}{0.19, 0.55, 0.91}
\definecolor{ao(english)}{rgb}{0.0, 0.5, 0.0}

\newcommand{\addcite}[0]{\ifthenelse{\boolean{showcomments}}
{\textcolor{purple}{(add cite(s)) }}{}}%

\newcommand{\enrique}[1]{  \ifthenelse{\boolean{showcomments}}
{\todo[inline,color=bleudefrance]{Enrique: #1}}{}}
\newcommand{\emmargin}[1]{\ifthenelse{\boolean{showcomments}}{\marginpar{\color{bleudefrance}\tiny EM: #1}}{}}

\newcommand{\Yan}[1]{  \ifthenelse{\boolean{showcomments}}
{\todo[inline,color=yellow]{Yan: #1}}{}}

\newcommand{\Petr}[1]{  \ifthenelse{\boolean{showcomments}}
{\todo[inline,color=green]{Petr: #1}}{}}

\newboolean{showedits}
\setboolean{showedits}{false}
\usepackage[markup=underlined]{changes}
\definechangesauthor[color=bleudefrance]{EM}
\newcommand{\aem}[1]{
\ifthenelse{\boolean{showedits}}
{\added[id=EM]{#1}}
{\!#1\hspace{-4.75pt}}
}
\newcommand{\repem}[2]{
\ifthenelse{\boolean{showedits}}
{\replaced[id=EM]{#1}{#2}}
{\!#1\hspace{-4.75pt}}
}
\newcommand{\dem}[1]{
\ifthenelse{\boolean{showedits}}
{\deleted[id=EM]{#1}}
{}
}

%% file: main.bbl
\begin{thebibliography}{10}
\providecommand{\url}[1]{#1}
\csname url@samestyle\endcsname
\providecommand{\newblock}{\relax}
\providecommand{\bibinfo}[2]{#2}
\providecommand{\BIBentrySTDinterwordspacing}{\spaceskip=0pt\relax}
\providecommand{\BIBentryALTinterwordstretchfactor}{4}
\providecommand{\BIBentryALTinterwordspacing}{\spaceskip=\fontdimen2\font plus
\BIBentryALTinterwordstretchfactor\fontdimen3\font minus
  \fontdimen4\font\relax}
\providecommand{\BIBforeignlanguage}[2]{{%
\expandafter\ifx\csname l@#1\endcsname\relax
\typeout{** WARNING: IEEEtran.bst: No hyphenation pattern has been}%
\typeout{** loaded for the language `#1'. Using the pattern for}%
\typeout{** the default language instead.}%
\else
\language=\csname l@#1\endcsname
\fi
#2}}
\providecommand{\BIBdecl}{\relax}
\BIBdecl

\bibitem{milano2018}
F.~Milano, F.~D\"orfler, G.~Hug, D.~J. Hill, and G.~Verbi{\v c}, ``Foundations
  and challenges of low-inertia systems (invited paper),'' in \emph{Proc. of
  Power Systems Computation Conference}, June 2018, pp. 1--25.

\bibitem{puschel2017mapping}
S.~P{\"u}schel-L, P.~Mancarella \emph{et~al.}, ``Mapping the frequency response
  adequacy of the australian national electricity market,'' in \emph{Proc. of
  Australasian Universities Power Engineering Conference}, 2017, pp. 1--6.

\bibitem{eirgrid2010all}
S.~EirGrid, ``All island tso facilitation of renewables studies,''
  \emph{Eirgrid, SONI, Dublin}, 2010.

\bibitem{o2014studying}
J.~O'Sullivan, A.~Rogers, D.~Flynn, P.~Smith, A.~Mullane, and M.~O'Malley,
  ``Studying the maximum instantaneous non-synchronous generation in an island
  system—frequency stability challenges in ireland,'' \emph{IEEE Transactions
  on Power Systems}, vol.~29, no.~6, pp. 2943--2951, Nov. 2014.

\bibitem{sanchez2019security}
F.~S{\'a}nchez, F.~Gonzalez-Longatt, and J.~L. Rueda, ``Security assessment of
  system frequency response,'' in \emph{Proc. of IEEE Power \& Energy Society
  General Meeting}, 2019, pp. 1--5.

\bibitem{greenwood2017frequency}
D.~Greenwood, K.~Y. Lim, C.~Patsios, P.~Lyons, Y.~S. Lim, and P.~Taylor,
  ``Frequency response services designed for energy storage,'' \emph{Applied
  Energy}, vol. 203, pp. 115--127, 2017.

\bibitem{eso2019technical}
N.~ESO, ``Technical report on the events of 9 august 2019,'' \emph{National
  Grid ESO}, 2019.

\bibitem{fang2017distributed}
J.~Fang, H.~Li, Y.~Tang, and F.~Blaabjerg, ``Distributed power system virtual
  inertia implemented by grid-connected power converters,'' \emph{IEEE
  Transactions on Power Electronics}, vol.~33, no.~10, pp. 8488--8499, 2017.

\bibitem{tielens2016relevance}
P.~Tielens and D.~Van~Hertem, ``The relevance of inertia in power systems,''
  \emph{Renewable and Sustainable Energy Reviews}, vol.~55, pp. 999--1009,
  2016.

\bibitem{tamrakar2017virtual}
U.~Tamrakar, D.~Shrestha, M.~Maharjan, B.~P. Bhattarai, T.~M. Hansen, and
  R.~Tonkoski, ``Virtual inertia: Current trends and future directions,''
  \emph{Applied Sciences}, vol.~7, no.~7, p. 654, 2017.

\bibitem{arani2012implementing}
M.~F.~M. Arani and E.~F. El-Saadany, ``Implementing virtual inertia in
  dfig-based wind power generation,'' \emph{IEEE Transactions on Power
  Systems}, vol.~28, no.~2, pp. 1373--1384, 2012.

\bibitem{almeida2015electric}
P.~R. Almeida, F.~J. Soares, and J.~P. Lopes, ``Electric vehicles contribution
  for frequency control with inertial emulation,'' \emph{Electric Power Systems
  Research}, vol. 127, pp. 141--150, 2015.

\bibitem{guggilam2018optimizing}
S.~S. Guggilam, C.~Zhao, E.~Dall’Anese, Y.~C. Chen, and S.~V. Dhople,
  ``Optimizing der participation in inertial and primary-frequency response,''
  \emph{IEEE Transactions on Power Systems}, vol.~33, no.~5, pp. 5194--5205,
  2018.

\bibitem{poolla2019placement}
B.~K. Poolla, D.~Gro{\ss}, and F.~D{\"o}rfler, ``Placement and implementation
  of grid-forming and grid-following virtual inertia and fast frequency
  response,'' \emph{IEEE Transactions on Power Systems}, vol.~34, no.~4, pp.
  3035--3046, 2019.

\bibitem{o1999modelling}
J.~O'Sullivan, M.~Power, M.~Flynn, and M.~O'Malley, ``Modelling of frequency
  control in an island system,'' in \emph{Power Engineering Society 1999 Winter
  Meeting, IEEE}, vol.~1.\hskip 1em plus 0.5em minus 0.4em\relax IEEE, 1999,
  pp. 574--579.

\bibitem{lalor2005impact}
G.~Lalor, J.~Ritchie, D.~Flynn, and M.~J. O'Malley, ``The impact of
  combined-cycle gas turbine short-term dynamics on frequency control,''
  \emph{IEEE Transactions on Power Systems}, vol.~20, no.~3, pp. 1456--1464,
  Aug. 2005.

\bibitem{ulbig2014impact}
A.~Ulbig, T.~S. Borsche, and G.~Andersson, ``Impact of low rotational inertia
  on power system stability and operation,'' \emph{IFAC Proceedings Volumes},
  vol.~47, no.~3, pp. 7290--7297, 2014.

\bibitem{NG2016}
\BIBentryALTinterwordspacing
``System operability framework 2016,'' National Grid Electricity System
  Operator, Tech. Rep., Nov. 2016. [Online]. Available:
  \url{https://www.nationalgrideso.com/document/63476/download}
\BIBentrySTDinterwordspacing

\bibitem{NG2016standard}
\BIBentryALTinterwordspacing
``National electricity transmission system security and quality of supply
  standard,'' National Grid Electricity System Operator, Tech. Rep., Apr. 2019.
  [Online]. Available:
  \url{https://www.nationalgrideso.com/document/141056/download}
\BIBentrySTDinterwordspacing

\bibitem{UCTLbook}
``Continental europe operation handbook: P1—policy 1: Load-fre- quency
  control and performance [{C}],'' Union for the Co-ordination of Transmission
  of Electricity, Tech. Rep., 2009.

\bibitem{knap2015sizing}
V.~Knap, S.~K. Chaudhary, D.-I. Stroe, M.~Swierczynski, B.-I. Craciun, and
  R.~Teodorescu, ``Sizing of an energy storage system for grid inertial
  response and primary frequency reserve,'' \emph{IEEE Transactions on Power
  Systems}, vol.~31, no.~5, pp. 3447--3456, Sept. 2015.

\bibitem{jiang2019dynamic}
Y.~Jiang, R.~Pates, and E.~Mallada, ``Dynamic droop control in low-inertia
  power systems,'' \emph{arXiv preprint: 1908.10983}, Aug. 2019.

\bibitem{chavez2014governor}
H.~Ch{\'a}vez, R.~Baldick, and S.~Sharma, ``Governor rate-constrained opf for
  primary frequency control adequacy,'' \emph{IEEE Transactions on Power
  Systems}, vol.~29, no.~3, pp. 1473--1480, 2014.

\end{thebibliography}
